\documentclass{article}

\usepackage{amsmath,amssymb,amsfonts,esint}

\newcommand\dps{\displaystyle }

 \usepackage{latexsym}
 \usepackage{pdflscape}
 \usepackage{examplep}
 \usepackage{longtable}
 \usepackage{graphicx}
 \usepackage{amsmath}
 \usepackage{booktabs}
 \usepackage{memhfixc}
 \usepackage{dsfont}
 \usepackage{amsmath,amssymb,amsfonts,amsthm}
 \usepackage[francais,english]{babel} 
 \newtheorem{theorem}{Theorem}[section]
 \newtheorem{proposition}{Proposition}[section]
 \newtheorem{remark}{Remark}[section]
 
 \newtheorem{corollary}{Corollary}[section]

\def\R{\mathbb{R}}
\def\N{\mathbb{N}}

\def\Z{\mathbb{Z}}
\def\P{\mathbb{P}}
\def\Vt{\widetilde{V}}
\def\cI{\mathcal{I}}
\def\cR{\mathcal{R}}

\title{Periodic Schr\"odinger operators with local defects \\
and spectral pollution\thanks{This work was financially supported by the ANR grant MANIF.}}

\author{Eric Canc\`es\thanks{Universit\'e Paris Est, CERMICS, Projet MICMAC, Ecole des Ponts ParisTech - INRIA, 6 \& 8 avenue Blaise Pascal, 77455 Marne-la-Vall\`ee Cedex 2, France, ({\tt cances@cermics.enpc.fr}, {\tt ehrlachv@cermics.enpc.fr})} \and Virginie Ehrlacher \and Yvon Maday\thanks{Universit\'e Pierre et Marie Curie-Paris 6, UMR 7598, Laboratoire J.-L. Lions, Paris, F-75005 France, and Division of Applied Mathematics, Brown University,
182 George Street, Providence, RI 02912, USA, ({\tt maday@ann.jussieu.fr})}}

\begin{document}

\maketitle

\begin{abstract} 
This article deals with the numerical calculation of eigenvalues of perturbed periodic Schr\"odinger operators located in spectral gaps. Such operators are encountered in the modeling of the electronic structure of crystals with local defects, and of photonic crystals. The usual finite element Galerkin approximation is known to give rise to spectral pollution. In this article, we give a precise description of the corresponding spurious states. We then prove that the supercell model does not produce spectral pollution. Lastly, we extend results by Lewin and S\'er\'e on some no-pollution criteria. In particular, we prove that using approximate spectral projectors enables one to eliminate spectral pollution in a given spectral gap of the reference periodic Schr\"odinger operator.
\end{abstract}




\section{Introduction}

\noindent 
Periodic Schr\"odinger operators are encountered in the modeling of the electronic structure of crystals, as well as the study of photonic crystals. They are self-adjoint operators on $L^2(\R^d)$ with domain $H^2(\R^d)$ of the form
$$
H^0_{\rm per} = -\Delta + V_{\rm per},
$$
where $\Delta$ is the Laplace operator and $V_{\rm per}$ a ${\mathcal R}$-periodic function of $L^p_{\rm loc}(\R^d)$ ($\mathcal R$ being a periodic lattice of $\R^d$), with $p=2$ if $d \le 3$, $p > 2$ for $d=4$ and $p=d/2$ for $d \ge 5$.  

Such operators describe perfect crystals, by contrast with real crystals, in which the underlying periodic structure is perturbed by the presence of local or extended defects. In solid state physics, local defects are due to impurities, vacancies, or interstitial atoms, while extended defects correspond to dislocations or grain boundaries. The properties of the crystal can be dramatically affected by the presence of defects. In this article, we consider the case of a $d$-dimensional crystal with a single local defect, whose properties are encoded in the perturbed periodic Schr\"odinger operator 
\begin{equation} \label{eq:PPSO}
H = H^0_{\rm per} + W = -\Delta + V_{\rm per} + W, \qquad W \in L^\infty(\R^d), \qquad W(x) \mathop{\rightarrow}_{|x| \to \infty} 0.
\end{equation}
Note that we do not assume here that $W$ is compactly supported. This allows us in particular to handle the mean-field model considered in~\cite{CDL}. In the latter model, $d=3$ and the self-consistent potential $W$ generated by the defect is of the form $W = \rho \star |\cdot|^{-1}$ with $\rho \in L^2(\R^3) \cap {\mathcal C}$, $\mathcal C$ denoting the Coulomb space. Such potentials are continuous and vanish at infinity, but are not compactly supported in general. 

\medskip

\noindent
Computing the spectrum of the operator $H$ is a key step to understand the properties of the system. It is well known that the self-adjoint operator $H^0_{\rm per}$ is bounded from below on  $L^2(\R^d)$, and that the spectrum $\sigma(H^0_{\rm per})$ of $H^0_{\rm per}$ is purely absolutely continuous, and composed of a finite or countable number of closed intervals of $\R$~\cite{ReedSimon4}. The open interval laying between two such closed intervals is called a spectral gap. The multiplication operator $W$ being a compact perturbation of $H^0_{\rm per}$, it follows from Weyl's theorem~\cite{ReedSimon4} that $H$ is self-adjoint on $L^2(\R^d)$ with domain $H^2(\R^d)$, and that $H$ and $H^0_{\rm per}$ have the same essential spectrum:
$$
\sigma_{\rm ess}(H) = \sigma_{\rm ess}(H^0_{\rm per}) = \sigma(H^0_{\rm per}).
$$
Contrarily to $H^0_{\rm per}$, which has no discrete spectrum, $H$ may possess discrete eigenvalues. While the discrete eigenvalues located below the minimum of $\sigma_{\rm ess}(H)$ are easily obtained by standard variational approximations (in virtue of the Rayleigh-Ritz theorem~\cite{ReedSimon4}), it is more difficult to compute numerically the discrete eigenvalues located in spectral gaps, for spectral pollution may occur~\cite{BoultonLevitin}.

\medskip

In Section~\ref{sec:spectral_pollution}, we recall that the usual finite element Galerkin approximation may give rise to spectral pollution~\cite{BoultonLevitin}, and give a precise description of the corresponding spurious states. In Section~\ref{sec:supercell}, we show that the supercell model does not produce spectral pollution. Lastly, we extend in Section~\ref{sec:projector} results by Lewin and S\'er\'e~\cite{LewinSere} on some no-pollution criteria, which guarantee in particular that the numerical method introduced in \cite{CDL}, involving approximate spectral projectors, and is spectral pollution free.

\newpage

\section{Galerkin approximation}
\label{sec:spectral_pollution}

The discrete eigenvalues of $H$ and the associated eigenvectors can be obtained by solving the variational problem
$$
\left\{
\begin{array}{l}
\mbox{find }(\psi, \lambda)\in H^1(\R^d) \times \R \mbox{ such that}\\
\forall \phi \in H^1(\R^d), \; a(\psi, \phi) = \lambda \langle   \psi, \phi \rangle_{L^2},\\
\end{array}
\right .
$$
where $\langle \cdot, \cdot \rangle_{L^2}$ is the scalar product of $L^2(\R^d)$ and $a$ the bilinear form associated with $H$:
$$
a(\psi,\phi) = \int_{\R^d} \nabla \psi \cdot \nabla \phi + \int_{\R^d} (V_{\rm per}+W) \psi\phi.
$$
A sequence $(X_n)_{n\in\N}$ of finite dimensional subspaces of 
$H^1(\R^d)$ being given, we consider for all $n\in\N$, the self-adjoint operator $H|_{X_n}: X_n \to X_n$ defined by
$$
\forall (\psi_n, \phi_n)\in X_n\times X_n, \; \langle H|_{X_n} \psi_n, \phi_n \rangle_{L^2} = a(\psi_n, \phi_n).
$$
The so-called Galerkin method consists in approximating the spectrum of the operator $H$ by the eigenvalues of the discretized operators $H|_{X_n}$ for $n$ large enough, the latter being obtained by solving the variational problem
\begin{equation}\label{eq:discdef}
\left\{
\begin{array}{l}
\mbox{find }(\psi_n, \lambda_n)\in X_n \times \R \mbox{ such that}\\
\forall \phi_n \in X_n, \; a(\psi_n, \phi_n) = \lambda_n \langle   \psi_n, \phi_n \rangle_{L^2}.\\
\end{array}
\right .
\end{equation}
According to the Rayleigh-Ritz theorem \cite{ReedSimon4}, under the natural assumption that the sequence $(X_n)_{n\in\N}$ satisfies
\begin{equation}\label{eq:density}
\forall \phi\in H^1(\R^d), \; \inf_{\phi_n\in X_n} \|\phi-\phi_n\|_{H^1} \mathop{\longrightarrow}_{n\to\infty} 0,
\end{equation}
the Galerkin method allows to compute the eigenmodes of $H$ associated with the discrete eigenvalues located below the bottom of the essential spectrum. It is also known (see e.g.~\cite{Chatelin} for details) that, as $H$ is bounded below, (\ref{eq:density}) implies
\begin{equation} \label{eq:no_lack}
\sigma(H) \subset \liminf_{n \to \infty} \sigma\left(H|_{X_n}\right),
\end{equation}
where the right-hand side is the limit inferior of the sets $\sigma\left(H|_{X_n}\right)$, that is the set of the complex numbers $\lambda$ such that there exists a sequence $(\lambda_n)_{n\in\N}$, with $\lambda_n \in \sigma(H|_{X_n})$ for each $n \in \N$, converging toward $\lambda$. In particular, any discrete eigenvalue $\lambda$ of the operator $H$ is well-approximated by a sequence of eigenvalues of the discretized operators $H|_{X_n}$. On the other hand, (\ref{eq:density}) is not strong enough an assumption to prevent spectral pollution. Some sequences of eigenvalues of $\sigma(H|_{X_n})$ may indeed converge to a real number which does not belong to the spectrum of $H$:
\begin{equation} \label{eq:spectral_pollution}
\limsup_{n \to \infty} \sigma\left(H|_{X_n}\right) \nsubseteq \sigma(H) \quad \mbox{in general},
\end{equation}
where the limit superior of the sets $\sigma\left(H|_{X_n}\right)$ is the set of the complex numbers $\lambda$ such that there exists a subsequence $(\sigma(H|_{X_{n_k}}))_{k \in \N}$ of $(\sigma(H|_{X_n}))_{n \in \N}$ for which 
$$
\forall k \in \N, \quad \exists \lambda_{n_k} \in \sigma(H|_{X_{n_k}})  \quad \mbox{and} \quad \lim_{k \to \infty} \lambda_{n_k} = \lambda. 
$$

Spectral pollution has been observed in many situations in physics and mechanics, and this phenomenon is now well-documented (see e.g. \cite{review_spectral_pollution} and references therein). In~\cite{BoultonLevitin}, Boulton and Levitin report numerical simulations on perturbed periodic Schr\"odinger operators showing that ``{\em the natural approach of truncating $\R^d$ to a large compact domain and applying the projection method to the corresponding Dirichlet problem is prone to spectral pollution}''. Truncating $\R^d$ indeed seems reasonable since it is known that the bound states of $H$ decay exponentially fast at infinity~\cite{MantoiuPurice}. The following result provides details on the behavior of the spurious modes when the approximation space is constructed using the finite element method.

\medskip

\begin{proposition} \label{prop:Galerkin} Let $({\mathcal T}_n^\infty)_{n \in \N}$ be a sequence of uniformly regular meshes of $\R^d$, invariant with respect to 
the translations of the lattice ${\mathcal R}$, and such that $h_n:=\max_{K \in {\mathcal T}_n^\infty} \mbox{\rm diam}(K) \mathop{\rightarrow}_{n \to \infty} 0$. 
Let $(\Omega_n)_{n \in \N}$ be an increasing sequence of closed convex sets of $\R^d$ converging to $\R^d$, 
${\mathcal T}_n := \left\{K \in {\mathcal T}_n^\infty \, | \, K \subset \Omega_n \right\}$ and $X_n$ the finite-dimensional approximation space of $H^1_0(\Omega_n) \hookrightarrow H^1(\R^d)$ obtained with ${\mathcal T}_n$ and $\P_m$ finite elements ($m \in \N^\ast$). Let $\lambda \in  \limsup_{n \to \infty} \sigma\left(H|_{X_n}\right) \setminus \sigma(H)$ and $(\psi_{n_k},\lambda_{n_k}) \in X_{n_k} \times \R$ be such that $H|_{X_{n_k}}\psi_{n_k} = \lambda_{n_k}\psi_{n_k}$,  $\|\psi_{n_k}\|_{L^2} = 1$ and $\lim_{k \to \infty}\lambda_{n_k} = \lambda$. Then, the sequence $(\psi_{n_k})_{k \in \N}$, considered as a sequence of functions of $H^1(\R^d)$, converges to $0$ weakly in $H^1(\R^d)$ and strongly in $L^q_{\rm loc}(\R^d)$,  with $q=\infty$ if $d=1$, $q < \infty$ if $d=2$ and $q < 2d/(d-2)$ if $d \ge 3$, in the sense that
$$
\forall K \subset \R^d, \quad K \mbox{ compact}, \quad \int_K|\psi_{n_k}|^q \mathop{\longrightarrow}_{k \to \infty} 0,
$$
and it holds  
\begin{equation}\label{eq:concentration}
\forall \epsilon > 0, \quad \exists R > 0 \quad \mbox{s. t.} \quad \liminf_{k \to \infty} \int_{\partial \Omega_{n_k} + B(0,R)} |\psi_{n_k}|^2 \ge 1-\epsilon.
\end{equation}
\end{proposition}

\medskip

The latter result shows that the mass of the spurious states concentrates on the boundary of the simulation domain $\Omega_{n_k}$. 

\medskip

This phenomenon is clearly observed on the two dimensional numerical simulations reported below, which have been performed with the finite element software FreeFem++~\cite{FreeFEM}, with $V_{\rm per}(x,y)=\cos(x) + 3\sin(2(x+y)+1)$ and $W(x,y)=-(x+2)^2(2y-1)^2\exp(-(x^2+y^2))$. We have checked numerically, using the Bloch decomposition method, that there is a gap $(\alpha,\beta)$, with $\alpha \simeq -0.341$ and $\beta \simeq 0.016$, between the first and second bands of $H^0_{\rm per}=-\Delta+V_{\rm per}$. We have also checked numerically, using the pollution free supercell method (see Theorem~\ref{Th:supercell} below), that $H=H^0_{\rm per}+W$ has exactly one eigenvalue in the gap $(\alpha,\beta)$ approximatively equal to $-0.105$. Our simulations have been performed with a sequence of $\P_1$-finite element approximation spaces $(X_n)_{40 \le n \le 100}$, where for each $40 \le n \le 100$,
\begin{itemize}
\item $\dps \Omega_n= \left[-4\pi \frac{m_n}n,4\pi \frac{m_n}n\right]$, with $\dps m_n=\left[n\left(\frac{n-40}{20}+5\right)\right]$; 
\item ${\cal T}_n^\infty$ is a uniform $2\pi\Z^2$-periodic mesh of $\R^2$ consisting of $2n^2$ isometrical isoceles rectangular triangles per unit cell.
\end{itemize}
The spectra of $H|_{X_n}$ in the gap $(\alpha,\beta)$ for $40 \le n \le 100$ are displayed on Fig.~1. We clearly see that all these operators have an eigenvalue close to $-0.1$, which is an approximation of a true eigenvalue of $H$. The corresponding eigenfunction for $n=88$ (blue circle on Fig.~1) is displayed on Fig.~\ref{fig:eigenfunction} (top); as expected, it is localized in the vicinity of the defect. On the other hand, most of these discretized operators have several eigenvalues in the range $(\alpha,\beta)$, which cannot be associated with an eigenvalue of $H$, and can be interpreted as spurious modes. The eigenfunction of $H|_{X_n}$ close to $-0.290$, obtained for $n=88$ (blue square on Fig.~1), is displayed on Fig.~\ref{fig:eigenfunction} (bottom); in agreement with the analysis carried out in Proposition~\ref{prop:Galerkin}, it is localized in the vicinity of the boundary of the computational domain.

\medskip

\begin{figure}[h]
\centering
\label{fig:spectrum}
\includegraphics[height=8truecm]{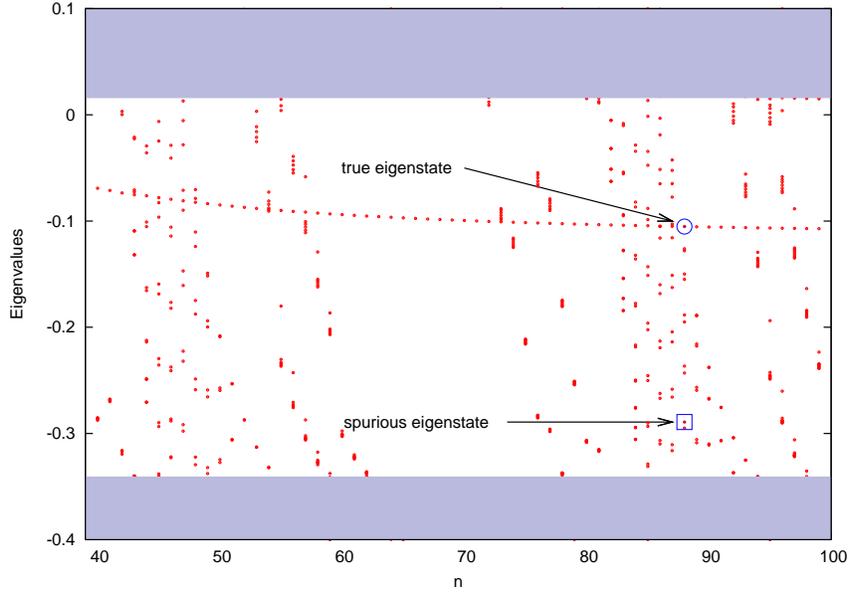}
\caption{Spectrum of $H|_{X_n}$ in the gap $(\alpha,\beta)$ for $40 \le n \le 100$}
\end{figure}

\medskip

\begin{figure}[h]
\centering
\label{fig:eigenfunction}
\begin{tabular}{c}
\includegraphics[height=6truecm]{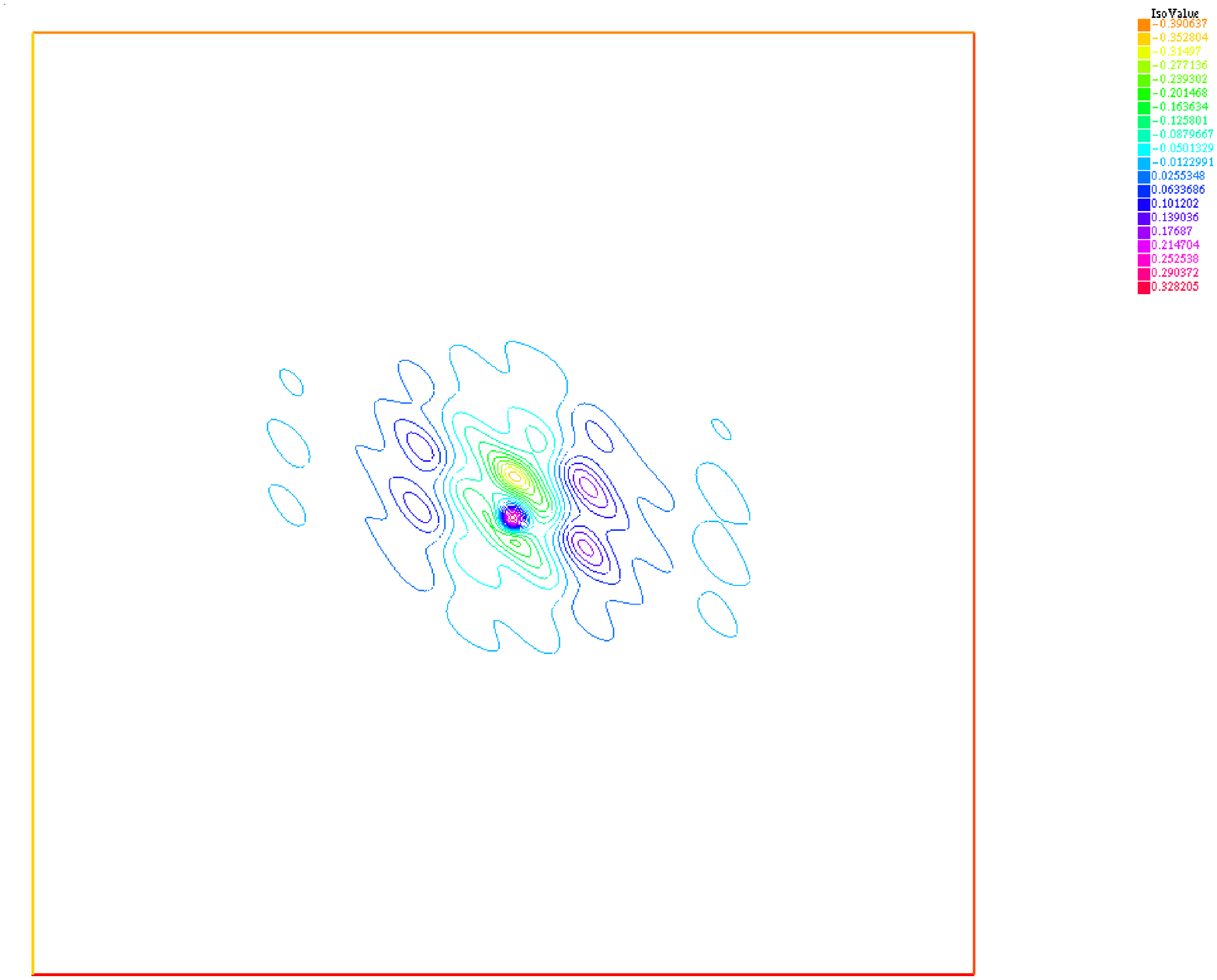} \\
\includegraphics[height=6truecm]{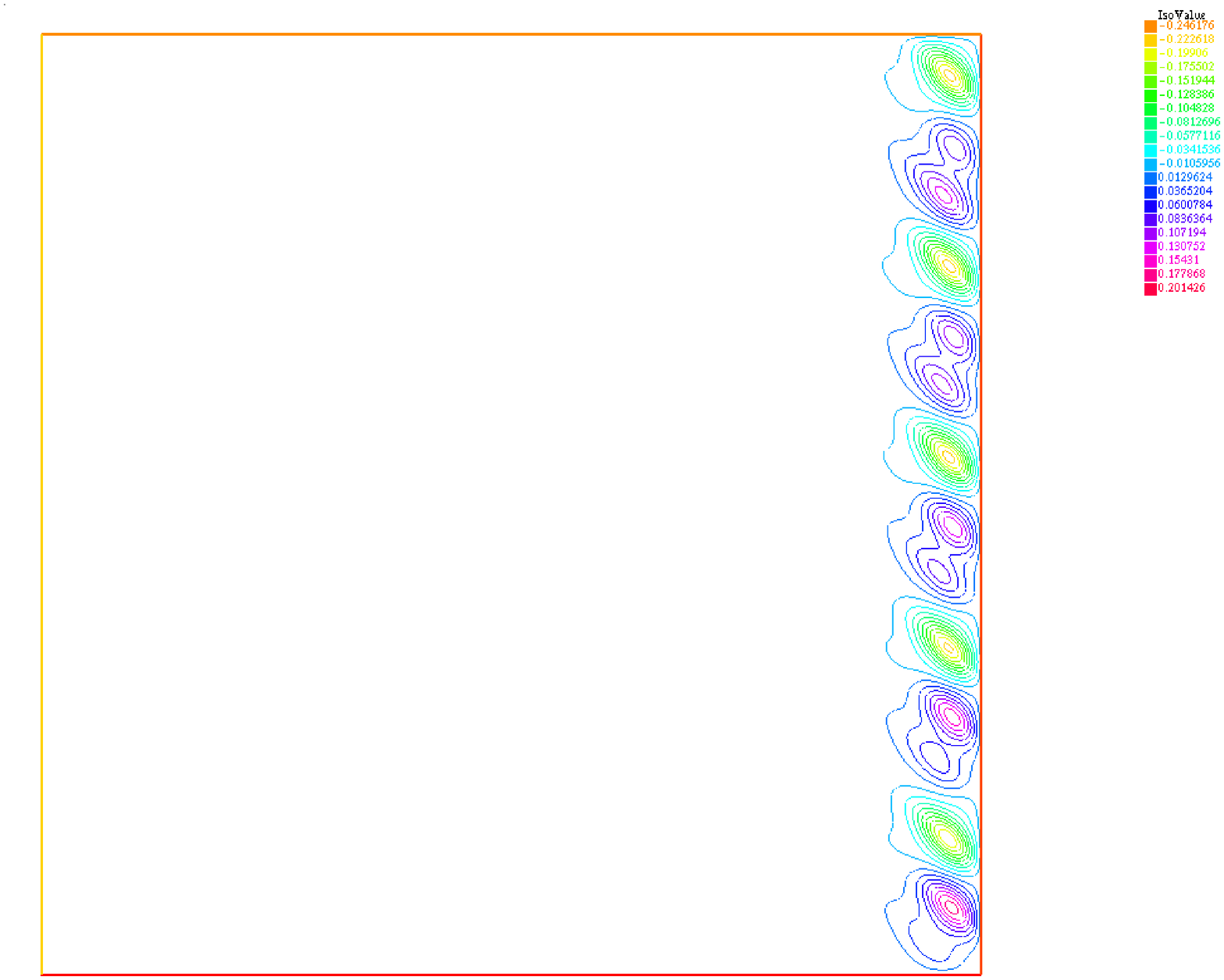}
\end{tabular}
\caption{A true eigenfunction, localized close to the defect (top), and a ``spurious'' eigenfunction, localized close to the boundary (bottom).}
\end{figure}

\medskip

\begin{remark}
Using the results in \cite{half-line}, it is possible to characterize the spurious states generated by finite element discretizations of one-dimensional perturbed Schr\"odinger operators: for ${\mathcal R}=b\Z$ and $\Omega_n = [-(n+t)b,(n+t)b]$, the spurious eigenvalues are the discrete eigenvalues in $[\min(\sigma(H^0_{\rm per})),+\infty) \setminus \sigma(H)$ of the operators $H^+(t)$ and $H^-(t)$ on $L^2(\R_+)$ with domains $H^2(\R_+) \cap H^1_0(\R_+)$, respectively defined by $H^\pm(t) = -\dps \frac{d^2}{dx^2} + V_{\rm per}(x\pm tb)$. Besides, the spurious eigenvectors of $H|_{X_n}$ converge (in some sense, and up to translation) to the discrete eigenvectors of $H^\pm(t)$. As 
$$
\left( \bigcup_{t \in [0,b)} \sigma(H^\pm(t)) \right) \cap  [\min(\sigma(H^0_{\rm per})),+\infty) = [\min(\sigma(H^0_{\rm per})),+\infty), 
$$
any $\lambda \in [\min(\sigma(H^0_{\rm per})),+\infty) \setminus \sigma(H)$ is a spurious eigenvalue, in the sense that there exists an increasing sequence $(\Omega_n)_{n \in \N}$ of closed intervals of $\R$ converging to $\R$ such that 
$$
\lambda \in \liminf_{n \to \infty} \sigma(H|_{X_n}).
$$
We refer to~\cite{these} for a proof and a numerical illustration of this result. The proof of similar results for $d \ge 2$ is work in progress.
\end{remark}

\medskip

\begin{proof}[Proof of Proposition~\ref{prop:Galerkin}] We first notice that, since $H = -\frac 12 \Delta + \frac 12 \left( -\Delta + 2V_{\rm per} \right) + W$, with $W$ bounded in $L^\infty(\R^d)$ and $-\Delta +2V_{\rm per}$ bounded below, there exists a constant $C \in \R_+$ such that
\begin{equation}\label{eq:lower_bound_a}
\forall \psi \in H^1(\R^d), \quad a(\psi,\psi) \ge \frac 12 \|\nabla \psi\|_{L^2}^2 - C \|\psi\|_{L^2}^2.
\end{equation}
As  
$$
\forall k \in \N, \quad \|\psi_{n_k}\|_{L^2} = 1 \quad \mbox{and} \quad a(\psi_{n_k},\psi_{n_k})= \lambda_{n_k} \mathop{\longrightarrow}_{k \to \infty} \lambda,
$$
we infer from (\ref{eq:lower_bound_a}) that the sequence $(\psi_{n_k})_{k \in \N}$ is bounded in $H^1(\R^d)$. It therefore converges, up to extraction, to some function $\phi \in H^1(\R^d)$, weakly in $H^1(\R^d)$, and strongly in $L^q_{\rm loc}(\R^d)$ with $q=\infty$ if $d=1$, $q < \infty$ if $d=2$ and $q < 2d/(d-2)$ if $d \ge 3$. It is easy to deduce from (\ref{eq:density}) and the continuity of $a$ on $H^1(\R^d)\times H^1(\R^d)$ that $\phi$ satisfies $H\phi = \lambda \phi$ and therefore that $\phi=0$ since $\lambda \notin \sigma(H)$ by assumption. Consequently, the whole sequence $(\psi_{n_k})_{k \in \N}$ converges to zero weakly in $H^1(\R^d)$ and strongly in $L^q_{\rm loc}(\R^d)$.

\medskip

Let us now prove (\ref{eq:concentration}) by contradiction. Assume that there exists $\epsilon > 0$ such that 
$$
\forall R > 0, \quad \liminf_{k \to \infty} \int_{\partial \Omega_{n_k} + B(0,R)}|\psi_{n_k}|^2 < 1-\epsilon.
$$
As $\|\psi_{n_k}\|_{L^2}=1$ for all $k$, the above inequality also reads
$$
\forall R > 0, \quad \limsup_{k \to \infty} \int_{\Omega_{n_k}^R}|\psi_{n_k}|^2 > \epsilon,
$$
where $\Omega_{n_k}^R = \left\{ x \in \Omega_{n_k} \, | \, d(x,\partial \Omega_{n_k}) \ge R \right\}$. We could then extract from $(\psi_{n_k})_{k \in \N}$ a subsequence, still denoted by $(\psi_{n_k})_{k \in \N}$, such that there exists an increasing sequence $(R_{n_k})_{k \in \N}$ of real numbers going to infinity such that
$$
\forall k \in \N, \quad \int_{\Omega_{n_k}^{R_{n_k}}}|\psi_{n_k}|^2 \ge \epsilon.
$$  
Let us denote by 
$$
C^0({\mathcal T}_n^\infty) = \left\{ v \in C^0(\R^d) \; | \; \forall K \in {\mathcal T}_n^\infty, \; v|_K \in \P_m \right\}
$$
the set of continuous functions built from ${\mathcal T}_n^\infty$ and $\P_m$-finite elements, and by
$$
X_n^\infty = C^0({\mathcal T}_n^\infty) \cap H^1(\R^d).
$$ 
The space $X_n^\infty$ is an (infinite dimensional) closed subspace of $H^1(\R^d)$. Obviously $X_n \hookrightarrow X_n^\infty$. We then introduce a sequence $(\chi_{n_k})_{k \in \N}$ of functions of $C^\infty_{\rm c}(\R^d)$ such that for all $k \in \N$,
$$
\mbox{\rm Supp}(\chi_{n_k}) \subset \Omega_{n_k}, \; \chi_k \equiv 1 \mbox{ on } \Omega_{n_k}^{R_{n_k}}, \; \mbox{and} \; \forall |\alpha| \le (m+1), \; \|\partial^\alpha \chi_{n_k}\|_{L^\infty} \le C R_{n_k}^{-|\alpha|},
$$
for a constant $C \in \R_+$ independent of $k$. Let $\widetilde \psi_{n_k} = P_{n_k} ( \chi_{n_k}\psi_{n_k})$, where $P_{n_k}$ is the interpolation projector on $X_{n_k}$. For all $k \in \N$, $\|\widetilde \psi_{n_k}\|_{L^2} \ge \epsilon^{1/2}$ and for all $\phi_{n_k}^\infty \in X_{n_k}^\infty$,
\begin{eqnarray*}
(a-\lambda_{n_k})(\widetilde \psi_{n_k},\phi_{n_k}^\infty) &=& (a-\lambda_{n_k})(\chi_{n_k}\psi_{n_k},\phi_{n_k}^\infty) \\ && - (a-\lambda_{n_k})(\chi_{n_k}\psi_{n_k}-P_{n_k}(\chi_{n_k} \psi_{n_k}),\phi_{n_k}^\infty) \\
&=& (a-\lambda_{n_k})(\psi_{n_k},\chi_{n_k}\phi_{n_k}^\infty) \\ && - (a-\lambda_{n_k})(\chi_{n_k}\psi_{n_k}-P_{n_k}(\chi_{n_k} \psi_{n_k}),\phi_{n_k}^\infty) \\
&& - \int_{\R^d} (\Delta\chi_{n_k} \psi_{n_k}\phi_{n_k}^\infty+2 \phi_{n_k}^\infty \nabla\chi_{n_k}\cdot\nabla\psi_{n_k}) \\
&=& (a-\lambda_{n_k})(\psi_{n_k},\chi_{n_k}\phi_{n_k}^\infty-P_{n_k}(\chi_{n_k} \phi_{n_k}^\infty)) \\
&& - (a-\lambda_{n_k})(\chi_{n_k}\psi_{n_k}-P_{n_k}(\chi_{n_k} \psi_{n_k}),\phi_{n_k}^\infty)    \\
&& - \int_{\R^d} (\Delta\chi_{n_k} \psi_{n_k}\phi_{n_k}^\infty+2 \phi_{n_k}^\infty \nabla\chi_{n_k}\cdot\nabla\psi_{n_k}),
\end{eqnarray*}
where we have used that $(a-\lambda_{n_k})(\psi_{n_k},P_{n_k}(\chi_{n_k} \phi_{n_k}^\infty))=0$ since $P_{n_k}(\chi_{n_k} \phi_{n_k}^\infty) \in X_{n_k}$. Denoting by 
$$
a^0(\psi,\phi) = \int_{\R^d}\nabla\psi\cdot\nabla\phi + \int_{\R^d} V_{\rm per}\psi\phi,
$$
we end up with
\begin{eqnarray}
(a^0-\lambda_{n_k})(\widetilde \psi_{n_k},\phi_{n_k}^\infty) &=& (a-\lambda_{n_k})(\psi_{n_k},\chi_{n_k}\phi_{n_k}^\infty-P_{n_k}(\chi_{n_k} \phi_{n_k}^\infty))\nonumber \\
&& - (a-\lambda_{n_k})(\chi_{n_k}\psi_{n_k}-P_{n_k}(\chi_{n_k} \psi_{n_k}),\phi_{n_k}^\infty)   \nonumber \\
&& - \int_{\R^d} (\Delta\chi_{n_k} \psi_{n_k}\phi_{n_k}^\infty+2 \phi_{n_k}^\infty \nabla\chi_{n_k}\cdot\nabla\psi_{n_k})\nonumber \\
&& - \int_{\R^d} W \widetilde \psi_{n_k} \phi_{n_k}^\infty. \label{eq:aml}
\end{eqnarray}

\medskip

\noindent
Besides, for $h_{n_k}\le 1$,
\begin{equation}\label{eq:interpolation}
\forall \phi_{n_k}^\infty \in X_{n_k}^\infty, \quad \|\chi_{n_k}\phi_{n_k}^\infty-P_{n_k}(\chi_{n_k}\phi_{n_k}^\infty)\|_{H^1} \le C h_{n_k}R_{n_k}^{-1} \|\phi_{n_k}^\infty\|_{H^1},
\end{equation}
for some constant $C$ independent of $k$ and $\phi_{n_k}^\infty$. To prove the above inequality, we notice that for all $K \in {\mathcal T}_{n_k}$, $(\chi_{n_k}\phi_{n_k}^\infty)|_K \in C^\infty(K)$, and $\partial^\beta\phi_{n_k}^\infty|_K=0$ if $|\beta|=m+1$, so that
\begin{eqnarray*}
\|\chi_{n_k}\phi_{n_k}^\infty-P_{n_k}(\chi_{n_k}\phi_{n_k}^\infty)\|_{H^1}^2 &=& 
\sum_{K \in {\mathcal T}_{n_k}} \| (\chi_{n_k}\phi_{n_k}^\infty)|_K-(P_{n_k}(\chi_{n_k}\phi_{n_k}^\infty))|_K\|_{H^1(K)}^2 \\ 
&\le& C h_{n_k}^{2m} \sum_{K \in {\mathcal T}_{n_k}} \max_{|\alpha|=m+1} \| \partial^\alpha(\chi_{n_k}\phi_{n_k}^\infty)|_K\|_{L^2(K)}^2 \\ 
&\le& C  h_{n_k}^{2m} \sum_{K \in {\mathcal T}_{n_k}} \max_{|\alpha|=m+1} \sum_{\beta \le \alpha} \| \partial^{\alpha-\beta}  \chi_{n_k}\|_{L^\infty}^2 \|\partial^\beta \phi_{n_k}^\infty|_K\|_{L^2(K)}^2  \\ 
&\le& C  h_{n_k}^{2m} R_{n_k}^{-2} \sum_{K \in {\mathcal T}_{n_k}} \max_{|\beta| \le m}  \|\partial^\beta \phi_{n_k}^\infty|_K\|_{L^2(K)}^2 \\
&\le& C  h_{n_k}^{2m} R_{n_k}^{-2} \sum_{K \in {\mathcal T}_{n_k}} (1+h_{n_k}^{-2(m-1)}) \|\phi_{n_k}^\infty|_K\|_{H^1(K)}^2 \\
&\le& Ch_{n_k}^2R_{n_k}^{-2} \|\phi_{n_k}^\infty\|_{H^1}^2,
\end{eqnarray*}
where we have used inverse inequalities and the assumption that the sequence of meshes $({\mathcal T}_n^\infty)_{n \in \N}$ is uniformly regular, to obtain the last but one inequality.

\medskip

\noindent
Using the boundedness of $(\psi_{n_k})_{k \in \N}$ in $H^1(\R^d)$, the properties of $\chi_{n_k}$ and $W$, and the fact that $(\psi_{n_k})_{k \in \N}$ strongly converges to $0$ in $L^2_{\rm loc}(\R^d)$, we deduce from (\ref{eq:aml}) and (\ref{eq:interpolation}) that
$$
\forall \phi_{n_k}^\infty \in X_{n_k}^\infty, \quad \left| (a^0-\lambda_{n_k})(\widetilde \psi_{n_k},\phi_{n_k}^\infty) \right| \le \eta_{n_k} \|\phi_{n_k}^\infty\|_{H^1},
$$
where the sequence of positive real numbers $(\eta_{n_k})_{k \in \N}$ goes to zero when $k$ goes to infinity.

\medskip

\noindent
We can now use Bloch theory (see e.g.~\cite{ReedSimon4}) and expand the functions of $X_{n_k}^\infty$ as 
$$
\phi_{n_k}^\infty(x) = \fint_{\Gamma^\ast} (\phi_{n_k}^\infty)_q(x) \, dq,
$$
where $\Gamma^\ast$ is the first Brillouin zone of the perfect crystal, and where for all $q \in \Gamma^\ast$,
$$
(\phi_{n_k}^\infty)_q (x) = \sum_{R \in {\mathcal R}} \phi_{n_k}^\infty(x+R) e^{-i q \cdot R}. 
$$
For each $q \in \Gamma^\ast$, the function $(\phi_{n_k}^\infty)_q$ belongs to the complex Hilbert space 
$$
L^2_q(\Gamma):= \left\{ v(x) e^{iq\cdot x}, \; v \in L^2_{\rm loc}(\R^d), \; v \mbox{ ${\mathcal R}$-periodic} \right\},
$$ 
where $\Gamma$ denotes the Wigner-Seitz cell of the lattice ${\mathcal R}$
(notice that the functions $(\phi_{n_k}^\infty)_q$ are complex-valued). Recall that if ${\mathcal R}=b \Z^d$ (cubic lattice of parameter $b > 0$), then $\Gamma=(-b/2,b/2]^d$ and $\Gamma^\ast = (-\pi/b,\pi/b]^d$. The mesh ${\mathcal T}_{n_k}^\infty$ being invariant with respect to the translations of the lattice ${\mathcal R}$, it holds in fact
$$
(\phi_{n_k}^\infty)_q \in C^0({\mathcal T}_{n_k}^\infty) \cap L^2_q(\Gamma).
$$
We thus have for all $\phi_{n_k}^\infty \in X_{n_k}^\infty$,
$$
(a^0-\lambda_{n_k})(\widetilde \psi_{n_k},\phi_{n_k}^\infty) 
= \fint_{\Gamma^\ast} (a^0_q-\lambda_n)((\widetilde \psi_{n_k})_q,(\phi_{n_k}^\infty)_q) \, dq,
$$
where
\begin{equation}\label{eq:a0q}
a^0_q(\psi_q,\phi_q) = \int_{\Gamma} \nabla\psi_q^\ast \cdot \nabla\phi_q + \int_\Gamma V_{\rm per} \psi_q^\ast \phi_q.
\end{equation}
Let $(\epsilon_{n,l,q},e_{n,l,q})_{1 \le l \le N_{n}}$, $\epsilon_{n,1,q} \le \epsilon_{n,2,q}  \le \cdots \le \epsilon_{n,N_n,q}$, be an $L^2_q(\Gamma)$-orthonormal basis of eigenmodes of $a^0_q$ in $C^0({\mathcal T}_n^\infty) \cap L^2_q(\Gamma)$. Expanding $(\widetilde \psi_{n_k})_q$ in the basis $(e_{n_k,l,q})_{1 \le l \le N_{n_k} }$, we get
$$
(\widetilde \psi_{n_k})_q = \sum_{j=1}^{N_{n_k}} c_{n_k,j,q} e_{n_k,j,q}.
$$
Choosing $\phi_{n_k}^\infty$ such that 
$$
(\phi_{n_k}^\infty)_q = \sum_{j=1}^{N_{n_k}} c_{n_k,j,q} (1_{\epsilon_{n_k,j,q}-\lambda_{n_k} \ge 0}-1_{\epsilon_{n_k,j,q}-\lambda_{n_k} < 0}) e_{n_k,j,q},
$$
we obtain $\|\phi_{n_k}^\infty\|_{L^2} = \|\widetilde \psi_{n_k}\|_{L^2}$ and 
$$
(a^0-\lambda_{n_k})(\widetilde \psi_{n_k},\phi_{n_k}^\infty) 
=   \fint_{\Gamma^\ast} \sum_{j=1}^{N_{n_k}} 
|\epsilon_{n_k,j,q}-\lambda_{n_k}| \, |c_{n_k,j,q}|^2.
$$
It is easy to check that $\dps \liminf_{k \to \infty}\max_{j,q} |\epsilon_{n_k,j,q}-\lambda_{n_k}| = \zeta:=\mbox{\rm dist}(\lambda,\sigma(H^0_{\rm per})) > 0$. Hence, 
$$
\liminf_{k \to \infty}(a^0-\lambda_{n_k})(\widetilde \psi_{n_k},\phi_{n_k}^\infty) 
\ge \zeta \epsilon.
$$
Besides, 
$$
\|\phi_{n_k}^\infty\|_{L^2}=\|\widetilde \psi_{n_k}\|_{L^2} \quad \mbox{and} \quad  a^0(\phi_{n_k}^\infty,\phi_{n_k}^\infty)=a^0(\widetilde \psi_{n_k},\widetilde \psi_{n_k}),
$$
which implies that the sequence $(\phi_{n_k}^\infty)_{k \in \N}$ is bounded in $H^1(\R^d)$. Consequently,
$$
0 < \zeta \epsilon \le \liminf_{k \to \infty}(a^0-\lambda_{n_k})(\widetilde \psi_{n_k},\phi_{n_k}^\infty) \le \liminf_{k \to \infty} \eta_{n_k} \|\phi_{n_k}^\infty\|_{H^1} = 0.
$$
We reach a contradiction.
\end{proof}

\medskip

A careful look on the above proof shows that the assumptions in Proposition~\ref{prop:Galerkin} can be weakened: in particular, the mesh ${\mathcal T}_n$ can be refined in the regions where $|W|$ is large, and coarsened in the vicinity of the boundary of $\Omega_n$ (see~\cite{these} for a more precise statement).

\section{Supercell method}
\label{sec:supercell}

In solid state physics and materials science, the current state-of-the-art technique to compute the discrete eigenvalues of a perturbed periodic Schr\"odinger operator in spectral gaps is the supercell method. Let ${\mathcal R}$ be the periodic lattice of the host crystal and $\Gamma$ its Wigner-Seitz cell. In the case of a cubic lattice of paramater $b > 0$, we have ${\mathcal R} = b \Z^d$ and $\Gamma = (-b/2,b/2]^d$. The supercell method consists in solving the spectral problem
\begin{equation}\label{eq:discdef2}
\left\{
\begin{array}{l}
\mbox{find }(\psi_{L,N}, \lambda_{L,N})\in X_{L,N} \times \R \mbox{ such that}\\
\forall \phi_{L,N} \in X_{L,N}, \; a_L(\psi_{L,N}, \phi_{L,N}) = \lambda_{L,N} \langle   \psi_{L,N}, \phi_{L,N} \rangle_{L^2_{\rm per}(\Gamma_L)}, \\
\end{array}
\right .
\end{equation}
where $\Gamma_L = L \Gamma$ (with $L \in \N^\ast$) is the supercell, 
$$
L^2_{\rm per}(\Gamma_L) = \left\{ u_L \in L^2_{\rm loc}(\R^d) \; | \; u_L \mbox{ $L{\mathcal R}$-periodic} \right\},
$$ 
$$
a_L(u_L,v_L) = \int_{\Gamma_L} \nabla u_L \cdot \nabla v_L + \int_{\Gamma_L} (V_{\rm per}+W) u_Lv_L, \quad \langle u_L,v_L\rangle_{L^2_{\rm per}(\Gamma_L)} = \int_{\Gamma_L} u_L v_L,
$$
and $X_{L,N}$ is a finite dimensional subspace of 
$$
H^1_{\rm per}(\Gamma_L) = \left\{ u_L \in L^2_{\rm per}(\Gamma_L) \; | \; \nabla u_L \in \left( L^2_{\rm per}(\Gamma_L) \right)^d\right\}.
$$
We denote by $H_{L,N} = H_L|_{X_{L,N}}$, where $H_L$ is the unique self-adjoint operator on $L^2_{\rm per}(\Gamma_L)$ associated with the quadratic form $a_L$. It then holds that 
$D(H_L) = H^2_{\rm per}(\Gamma_L)$, 
$$
\forall \phi_L \in H^2_{\rm per}(\Gamma_L), \quad H_L\phi_L = -\Delta \phi_L + (V_{\rm per} + W_L)\phi_L,
$$
and
$$
\forall \phi_{L,N} \in X_{L,N}, \quad H_{L,N}\phi_{L,N} = -\Delta \phi_{L,N} + \Pi_{X_{L,N}}\left( (V_{\rm per} + W_L)\phi_{L,N}\right),
$$
where $W_L\in L^{\infty}_{\rm per}(\Gamma_L)$ denotes the $L\cR$-periodic extension of $W|_{\Gamma_L}$ and $\Pi_{X_{L,N}}$ is the orthogonal projector of $L^2_{\rm per}(\Gamma_L)$ 
on $X_{L,N}$ for the $L^2_{\rm per}(\Gamma_L)$ inner product.

Again for the sake of clarity, we restrict ourselves to cubic lattices (${\mathcal R} = b \Z^d$) and to the most popular discretization method for supercell model, namely the Fourier (also called planewave) method. We therefore consider approximation spaces of the form
$$
X_{L,N} = \left\{ \sum_{k \in 2\pi (bL)^{-1}\Z^d \, | \, |k| \le 2\pi (bL)^{-1}N} c_k e_{L,k} \; \big| \; \forall k, \, c_{-k}=c_k^\ast \right\},
$$
where $e_{L,k}(x) = |\Gamma_L|^{-1/2} e^{ik \cdot x}$.

From the classical Jackson inequality for Fourier truncation, we deduce by scaling the following property of the discretization spaces $X_{L,N}$: for all real numbers $r$ and $s$ such that $0\leq r \leq s$, there exists a constant 
$C>0$ such that for all $L\in\N^*$ and all $\phi_L\in H^s_{\rm per}(\Gamma_L)$, 
\begin{equation}\label{eq:regularity}
\|\phi_L - \Pi_{X_{L,N}}\phi_L\|_{H^r_{\rm per}(\Gamma_L)} \leq C \left(\frac{L}{N} \right)^{s-r} \|\phi_L\|_{H^s_{\rm per}(\Gamma_L)}.
\end{equation}

\medskip

\noindent
Our analysis of the supercell method requires some assumption on the potential $V_{\rm  per}$. We define the functional space ${\cal M}_{\rm per}(\Gamma)$ as
$$
{\cal M}_{\rm per}(\Gamma) =\left\{ v\in L^2_{\rm per}(\Gamma) \;  | \; \|v\|_{{\cal M}_{\rm per}(\Gamma)}:=\sup_{L\in \N^\ast} \sup_{w\in H^1_{\rm per}(\Gamma_L)\setminus\left\{0\right\}}\frac{\| vw\|_{L^2_{\rm per}(\Gamma_L)}}{\| w \|_{H^1_{\rm per}(\Gamma_L)}} < \infty\right\}.
$$
It is quite standard to prove that ${\cal M}_{\rm per}(\Gamma)$ is a normed space and that the space of the ${\mathcal R}$-periodic functions of class $C^\infty$ is dense in ${\cal M}_{\rm per}(\Gamma)$. We denote the ${\cal R}$-periodic Lorentz spaces~\cite{BL} by $L^{p,q}_{\rm per}(\Gamma)$.

\medskip

\begin{proposition} The following embeddings are continuous:
\begin{eqnarray*}
&& {\rm for } \ d=1,\quad L^{2}_{\rm per}(\Gamma) \hookrightarrow {\cal M}_{\rm per}(\Gamma), \\
&& {\rm for } \ d=2,\quad L^{2,\infty}_{\rm per}(\Gamma) \hookrightarrow {\cal M}_{\rm per}(\Gamma), \\
&& {\rm for } \ d=3,\quad L^{3,\infty}_{\rm per}(\Gamma) \hookrightarrow {\cal M}_{\rm per}(\Gamma ).
\end{eqnarray*}
\end{proposition}

\medskip

\begin{proof} We only prove the result for $d=3$; the other two embeddings are obtained by similar arguments. Let us first recall that the Lorentz space  $L^{3,\infty}(\Gamma)$ is a $L^2$-multiplier of $L^{6,2}(\Gamma)$ (this can be seen by combining results on convolution multiplier spaces \cite{Avci} and continuity properties of the Fourier transform on Lorentz spaces \cite{BL}), in the sense that
\begin{equation*}
\exists C_1\in\R_+ \; | \;   \forall f\in L^{3,\infty}(\Gamma),\; \forall g\in L^{6,2}(\Gamma),\; \| fg\|_{L^2(\Gamma)}\le C_1 \| f\|_{L^{3,\infty}(\Gamma)} \| g\| _{L^{6,2}(\Gamma)}.
\end{equation*}
Besides, the embedding of $H^1(\Gamma)$ into $L^{6,2}(\Gamma)$ is continuous (see~\cite{Alvino} for instance)
\begin{equation}
\exists C_2\in\R_+ \; | \;  \forall g \in H^1(\Gamma), \; \| g \|_{L^{6,2}(\Gamma)} \le C_2  \| g \|_{H^{1}(\Gamma)}.
\end{equation}
Let $v\in L^{3,\infty}_{\rm per}(\Gamma)$. Denoting by $\cI_L := {\cal R} \cap (-Lb/2,Lb/2]^3$, we have, for all $w\in H^1_{\rm per}(\Gamma_L)$,
\begin{eqnarray*}
\| vw\|^2_{L^2_{\rm per}(\Gamma_L)}&=&\int_{\Gamma_L}|v w|^2  =  \sum_{R\in\cI_L} \int_{\Gamma+R}|v(x) w(x)|^2 \, dx \\
&=&   \sum_{R\in\cI_L} \int_{\Gamma} |v(x) w(x+R)|^2 \, dx =  \sum_{R\in\cI_L} \|v  w(.+R)\|_{L^2(\Gamma)}^2
\\
& \leq &  C_1^2 \sum_{R\in\cI_L} \|v\|^2_{L^{3,\infty}(\Gamma)} \|w(.+R)\|_{L^{6,2}(\Gamma)}^2\\
&\leq & C_1^2  \|v\|^2_{L^{3,\infty}(\Gamma)} \sum_{R\in\cI_L} \|w(.+R)\|_{L^{6,2}(\Gamma)}^2\\
&\leq & C_1^2C_2^2  \|v\|^2_{L^{3,\infty}(\Gamma)} \sum_{R\in\cI_L} \|w(.+R)\|_{H^{1}(\Gamma)}^2 \\
&\leq &  C_1^2C_2^2  \|v\|^2_{L^{3,\infty}(\Gamma)} 
\sum_{R\in\cI_L} \int_{\Gamma} \left(|w(x+R)|^2 + |\nabla w(x+R)|^2\right)\,dx \\
&\leq &  C_1^2C_2^2  \|v\|^2_{L^{3,\infty}(\Gamma)}   \int_{\Gamma_L} \left(|w(x)|^2 + |\nabla w(x)|^2\right)\,dx \\
&\leq &  C_1^2C_2^2  \|v\|^2_{L^{3,\infty}(\Gamma)}   \|w\|^2_{H^1_{\rm per}(\Gamma_L)}.
\end{eqnarray*}
Therefore, $v \in {\cal M}_{\rm per}(\Gamma)$ and $\|v\|_{{\cal M}_{\rm per}(\Gamma)} \le C_1C_2 \|v\|_{L^{3,\infty}(\Gamma)}$.
\end{proof}

\medskip

\begin{remark} In dimension 3, the ${\mathcal R}$-periodic Coulomb kernel $G_1$ defined by
$$
-\Delta G_1 = 4\pi \left( \sum_{R \in {\cal R}} \delta_R - |\Gamma|^{-1} \right), \quad \min_{x \in \R^3} G_1(x)=0,
$$
is in $L^{3,\infty}_{\rm per}(\Gamma)$, hence in ${\cal M}_{\rm per}(\Gamma)$. The functional setting we have introduced therefore allows us to deal with the electronic structure of crystals containing point-like nuclei.
\end{remark}

\medskip

\begin{theorem} \label{Th:supercell} Assume that  $V_{\rm per} \in {\cal M}_{\rm per}(\Gamma)$. Then
$$
\lim_{N,L \to \infty \, | \, N/L \to \infty} \sigma(H_{L,N}) = \sigma(H).
$$
\end{theorem}

\medskip

\begin{proof} Let us first establish that
$$
\sigma(H) \subset \liminf_{N,L \to \infty \, | \, N/L \to \infty} \sigma(H_{L,N}).
$$
Let $\lambda \in \sigma(H)$ and $(N_L)_{L\in\N^*}$ be a sequence of integers such that $\dps \frac{N_L}{L} \mathop{\longrightarrow}_{L\to \infty} \infty$. Let $\epsilon > 0$ and $\psi \in C^\infty_{\rm c}(\R^d)$ be such that $\|\psi\|_{L^2}=1$ and $\|(H-\lambda)\psi\|_{L^2} \le \epsilon$. We denote by $\psi_L$ the $L{\cal R}$-periodic extension of $\psi|_{\Gamma_L}$. Since $\psi$ is compactly supported, there exists $L_0 \in \N^\ast$ such that for all $L \ge L_0$, $\mbox{Supp}(\psi) \subset \Gamma_L$. Consequently, for all $L \ge L_0$, 
$$
\|\psi_L\|_{L^2_{\rm per}(\Gamma_L)}=1 \quad \mbox{and} \quad 
\|(H_L-\lambda)\psi_L\|_{L^2_{\rm per}(\Gamma_L)} \le \epsilon.
$$
Let $\psi_{L,N_L}:= \Pi_{X_{L,N_L}}\psi_L$. We are going to prove that 
\begin{equation}\label{eq:convHL}
\left\|(H_L - \lambda)\psi_L - \left( H_{L,N_L} - \lambda\right) \psi_{L,N_L} \right\|_{L^2_{\rm per}(\Gamma_L)} \mathop{\longrightarrow}_{L\to\infty} 0.
\end{equation}
First, we infer from (\ref{eq:regularity}) and the density of $H^1_0(\Omega)$ in $L^2(\Omega)$ for any bounded domain $\Omega$ of $\R^d$, that
\begin{eqnarray*}
\forall \phi \in L^2_{\rm c}(\R^d), \quad \|(1-\Pi_{X_{L,N_L}})\phi_L\|_{L^2_{\rm per}(\Gamma_L)} \mathop{\longrightarrow}_{L \to \infty} 0,
\end{eqnarray*}
where $L^2_{\rm c}(\R^d)$ denotes the space of the square integrable functions on $\R^d$ with compact supports, and where $\phi_L$ is the $L{\cal R}$-periodic extension of $\phi|_{\Gamma_L}$. As $\psi$, $\Delta\psi$, $V_{\rm per}\psi$ and $W\psi$ are square integrable, with compact supports, we therefore have for all $L \ge L_0$,
\begin{eqnarray*}
&&\|\psi_L - \psi_{L,N_L}\|_{L^2_{\rm per}(\Gamma_L)} =  \left\|\left(1 - \Pi_{X_{L,N_L}}\right)\psi_L\right\|_{L^2_{\rm per}(\Gamma_L)} \mathop{\longrightarrow}_{L\to\infty}  0, \\
&&\|-\Delta \psi_L +\Delta \psi_{L,N_L}\|_{L^2_{\rm per}(\Gamma_L)} =  \left\|\left(1 - \Pi_{X_{L,N_L}}\right)(-\Delta \psi)_L\right\|_{L^2_{\rm per}(\Gamma_L)} \mathop{\longrightarrow}_{L\to\infty}  0, \\
&&\|W_L\psi_L - \Pi_{X_{L,N_L}} (W_L \psi_L) \|_{L^2_{\rm per}(\Gamma_L)} =  \left\|\left( 1 -\Pi_{X_{L,N_L}} \right) (W\psi)_L \right\|_{L^2_{\rm per}(\Gamma_L)} \mathop{\longrightarrow}_{L\to\infty}  0, \\
&& \|V_{\rm per}\psi_L - \Pi_{X_{L,N_L}}(V_{\rm per}\psi_L) \|_{L^2_{\rm per}(\Gamma_L)} =  \left\|\left( 1 -\Pi_{X_{L,N_L}} \right) (V_{\rm per}\psi)_L \right\|_{L^2_{\rm per}(\Gamma_L)} \mathop{\longrightarrow}_{L\to\infty}  0.
\end{eqnarray*}
We infer from the last two convergence results that, on the one hand,
\begin{eqnarray*}
&&\|W_L\psi_L - \Pi_{X_{L,N_L}} (W_L \psi_{L,N_L}) \|_{L^2_{\rm per}(\Gamma_L)}\\
&& \quad \leq  \left\|W_L\psi_L - \Pi_{X_{L,N_L}} (W_L \psi_L) \right\|_{L^2_{\rm per}(\Gamma_L)} + \left\|\Pi_{X_{L,N_L}}\left(W_L(\psi_L - \psi_{L,N_L})\right) \right\|_{L^2_{\rm per}(\Gamma_L)}\\
&&  \quad \leq \left\|W_L\psi_L - \Pi_{X_{L,N_L}} (W_L \psi_L)\right\|_{L^2_{\rm per}(\Gamma_L)}+ \|W\|_{L^{\infty}}\left\|\psi_L - \psi_{L,N_L} \right\|_{L^2_{\rm per}(\Gamma_L)} \\
&& \qquad \mathop{\longrightarrow}_{L\to\infty}  0,
\end{eqnarray*}
and that, on the other hand,
\begin{eqnarray*}
&& \|V_{\rm per}\psi_L - \Pi_{X_{L,N_L}}(V_{\rm per}\psi_{L,N_L}) \|_{L^2_{\rm per}(\Gamma_L)} \\
&& \quad \leq  \left\|V_{\rm per}\psi_L - \Pi_{X_{L,N_L}}(V_{\rm per}\psi_{L}) \right\|_{L^2_{\rm per}(\Gamma_L)} + \left\|\Pi_{X_{L,N_L}}\left(V_{\rm per} (\psi_L - \psi_{L,N_L})\right)\right\|_{L^2_{\rm per}(\Gamma_L)}\\
&& \quad \leq  \left\| \left( 1 - \Pi_{X_{L,N_L}}\right) V_{\rm per}\psi_L \right\|_{L^2_{\rm per}(\Gamma_L)} + \|V_{\rm per}\|_{{\cal M}_{\rm per}(\Gamma)} \|\psi_L - \psi_{L,N_L} \|_{H^1_{\rm per}(\Gamma_L)} \\ 
&& \qquad \mathop{\longrightarrow}_{L\to\infty}  0.
\end{eqnarray*}
Collecting the above results, we obtain (\ref{eq:convHL}). Thus, for $L$ large enough,
$$
\|(H_{L,N_L}- \lambda)\psi_{L,N_L}\|_{L^2_{\rm per}(\Gamma_L)} \leq 2\varepsilon.
$$
As $\|\psi_{L,N_L}\|_{L^2_{\rm per}(\Gamma_L)}=1$ for all $L \ge L_0$, we infer that for $L$ large enough, $\mbox{dist}(\lambda,\sigma(H_{L,N_L}))\le 2\epsilon$, so that $\dps \lambda \in  \liminf_{L \to \infty} \sigma(H_{L,N_L})$.

\medskip

\noindent
Let us now prove that 
$$
\limsup_{N,L \to \infty \, | \, N/L \to \infty} \sigma(H_{N,L}) \subset \sigma(H).
$$
We argue by contradiction, assuming that there exists $\lambda \in \R \setminus \sigma(H)$ and a sequence $(L_k,N_k)_{k \in \N}$ with $\dps L_k \mathop{\rightarrow}_{k \to \infty} \infty$,  $\dps N_k \mathop{\rightarrow}_{k \to \infty} \infty$,  $\dps N_k/L_k \mathop{\rightarrow}_{k \to \infty} \infty$, such that for each $k$, there exists $(\psi_{L_k,N_k},\lambda_{L_k,N_k}) \in X_{L_k,N_k} \times \R$ satisfying
\begin{equation*}
\left\{
\begin{array}{l}
\forall \phi_{L_k,N_k} \in X_{L_k,N_k}, \; a_{L_k}(\psi_{L_k,N_k}, \phi_{L_k,N_k}) = \lambda_{L_k,N_k} \langle   \psi_{L_k,N_k}, \phi_{L_k,N_k} \rangle_{L^2_{\rm per}(\Gamma_{L_k})} \\
\|\psi_{L_k,N_k}\|_{L^2(\Gamma_{L_k})} = 1,
\end{array}
\right .
\end{equation*}
and $\dps \lim_{k \to \infty} \lambda_{L_k,N_k}=\lambda$. Each function $\psi_{L_k,N_k}$ is then solution to the PDE
\begin{equation}\label{eq:psikeq}
- \frac 12 \Delta \psi_{L_k,N_k} + \Pi_{X_{L_k,N_k}} \left((V_{\rm per}+W_{L_k}) \psi_{L_k,N_k}\right) = \lambda_{L_k,N_k} \psi_{L_k,N_k}.
\end{equation}
Reasoning as in the proof of Proposition~\ref{prop:Galerkin}, it can be checked that the sequence $(\|\psi_{L_k,N_k}\|_{H^1_{\rm per}(\Gamma_{L_k})})_{k \in \N}$ is bounded, and that
\begin{equation}\label{eq:psiLkNk}
\psi_{L_k,N_k} \mathop{\longrightarrow}_{k \to \infty} 0  \quad \mbox{in } L^2_{\rm loc}(\R^d).
\end{equation}

\medskip

\noindent
For all $k$, we consider a cut-off function $\chi_k \in C^\infty_{\rm c}(\R^d)$ such that $0 \le \chi_k \le 1$ on $\R^d$, $\chi_k \equiv 1$ on $\Gamma_{L_k}$, $\mbox{\rm Supp}(\chi_k)  \subset (L_k+L_k^{1/2})\Gamma$, $\|\nabla \chi_k\|_{L^\infty} \le C L_k^{-1/2}$, and $\|\Delta \chi_k\|_{L^\infty} \le C L_k^{-1}$ for some constant $C \in \R_+$ independent of $k$. We then set $\widetilde \psi_k = \chi_k \psi_{L_k,N_k}$. It holds $\widetilde \psi_k \in H^2(\R^d)$, $1 \le \|\widetilde \psi_k\|_{L^2} \le 2^{d/2}$ and 
\begin{eqnarray}
- \frac 12 \Delta \widetilde \psi_k + V_{\rm per} \widetilde \psi_k - \lambda \widetilde \psi_k &=& \chi_k \left( V_{\rm per} \psi_{L_k,N_k} - \Pi_{X_{L_k,N_k}}\left(V_{\rm per} \psi_{L_k,N_k}\right)\right) \nonumber \\ && 
- \chi_k \Pi_{X_{L_k,N_k}}\left(W_{L_k} \psi_{L_k,N_k}\right)   - \nabla\chi_k \cdot \nabla\psi_{L_k,N_k}  \nonumber \\ && - \frac 12 \Delta\chi_k \psi_{L_k,N_k} + (\lambda_{L_k,N_k}-\lambda)\widetilde \psi_k . \label{eq:tildepsi}
\end{eqnarray}
As $(\lambda_{L_k,N_k})_{k \in \N}$ converges to $\lambda$ in $\R$ and $\|\widetilde \psi_k\|_{L^2} \le 2^{d/2}$, we have
$$
(\lambda_{L_k,N_k}-\lambda) \widetilde \psi_k \mathop{\longrightarrow}_{k \to \infty} 0 \quad \mbox{strongly in } L^2(\R^d).
$$
Using the facts that $\mbox{Supp}(\chi_k) \subset 2\Gamma_{L_k}$, $\|\nabla \chi_k\|_{L^\infty} \le C L_k^{-1/2}$ and $\|\Delta \chi_k\|_{L^\infty} \le C L_k^{-1}$ for a constant $C \in \R_+$ independent of $k$, and the boundedness of the sequence $(\|\psi_{L_k,N_k}\|_{H^1_{\rm per}(\Gamma_{L_k})})_{k \in \N}$, we get
$$
- \nabla\chi_k \cdot \nabla\psi_{L_k,N_k} - \frac 12 \Delta\chi_k \psi_{L_k,N_k} \mathop{\longrightarrow}_{k \to \infty} 0 \quad \mbox{strongly in } L^2(\R^d).
$$
It also follows from (\ref{eq:psiLkNk}) that the sequence $\| W_{L_k} \psi_{L_k,N_k}\|_{L^2_{\rm per}(\Gamma_{L_k})} $ goes to zero, leading to
$$
\chi_{k}\Pi_{X_{L_k,N_k}}\left(W_{L_k} \psi_{L_k,N_k}\right)\mathop{\longrightarrow}_{k \to \infty} 0 \quad \mbox{strongly in } L^2(\R^d).
$$ 
Lastly, 
\begin{equation}\label{eq:conv}
\chi_k \left( V_{\rm per} \psi_{L_k,N_k} - \Pi_{X_{L_k,N_k}}\left(V_{\rm per} \psi_{L_k,N_k}\right)  \right) \mathop{\longrightarrow}_{k \to \infty} 0 \quad \mbox{strongly in } L^2(\R^d).
\end{equation}
To show the above convergence result, we consider $\epsilon > 0$ and, using the density of e.g.  $W^{1,\infty}_{\rm per}(\Gamma):= \{ W_{\rm per} \in L^\infty_{\rm per}(\Gamma) \; | \; \nabla W_{\rm per} \in L^{\infty}_{\rm per}(\Gamma)\}$ 
in ${\cal M}_{\rm per}(\Gamma)$, we can choose some $\Vt_{\rm per} \in W^{1,\infty}_{\rm per}(\Gamma)$ such that $\|V_{\rm per} - \Vt_{\rm per}\|_{{\cal M}_{\rm per}(\Gamma)} \leq \varepsilon$. 
We then deduce from (\ref{eq:regularity}) that, for all $k\in\N$,
\begin{eqnarray*}
&&\left\|V_{\rm per} \psi_{L_k,N_k} - \Pi_{X_{L_k,N_k}}\left(V_{\rm per} \psi_{L_k,N_k}\right)\right\|_{L^2_{\rm per}(\Gamma_{L_k})} \\
&& \leq \left\|(V_{\rm per} - \Vt_{\rm per}) \psi_{L_k,N_k}\right\|_{L^2_{\rm per}(\Gamma_{L_k})} +  \left\|\Vt_{\rm per} \psi_{L_k,N_k} - \Pi_{X_{L_k,N_k}}\left(\Vt_{\rm per} \psi_{L_k,N_k}\right)\right\|_{L^2_{\rm per}(\Gamma_{L_k})} \\
&& \leq  \|V_{\rm per} - \Vt_{\rm per}\|_{{\cal M}_{\rm per}(\Gamma)} \|\psi_{L_k,N_k}\|_{H^1_{\rm per}(\Gamma_{L_k})} + \frac{L_k}{N_k} \|\Vt_{\rm per} \psi_{L_k,n_k}\|_{H^1_{\rm per}(\Gamma_{L_k})}\\
&& \leq  \varepsilon \|\psi_{L_k,N_k}\|_{H^1_{\rm per}(\Gamma_{L_k})} + \frac{L_k}{N_k} \|\psi_{L_k,n_k}\|_{H^1_{\rm per}(\Gamma_{L_k})}(\|\Vt_{\rm per}\|_{L^{\infty}} + \|\nabla \Vt_{\rm per}\|_{L^{\infty}}).\\
\end{eqnarray*}
Since the sequence $\left( \|\psi_{L_k,N_k}\|_{H^1_{\rm per}(\Gamma_{L_k})}\right)_{k\in\N^*}$ is bounded, this yields
$$
\left\|  V_{\rm per} \psi_{L_k,N_k} - \Pi_{X_{L_k,N_k}}\left(V_{\rm per} \psi_{L_k,N_k}\right) \right\|_{L^2_{\rm per}(\Gamma_{L_k})} \mathop{\longrightarrow}_{k\to\infty} 0,
$$
which implies (\ref{eq:conv}).

\medskip

\noindent
Collecting the above convergence results, we obtain that the right-hand side of (\ref{eq:tildepsi}) goes to zero strongly in $L^2(\R^d)$. Therefore, $(\widetilde \psi_k/\|\widetilde \psi_k\|_{L^2})_{k \in \N}$ is a Weyl sequence for $\lambda$, which contradicts the fact that $\lambda \notin \sigma(H^0_{\rm per})$.
\end{proof}

\medskip

\noindent
A similar result was proved in~\cite{Soussi} for compactly supported defects in 2D photonic crystals, with $V_{\rm per} \in L^\infty(\R^2)$ and $N=\infty$. In~\cite{CEM}, we prove that the error made on the eigenvalues and the associated eigenvectors decays exponentially with respect to the size of the supercell. We did not consider here the error due to numerical integration. The numerical analysis of the latter is ongoing work and will be reported in \cite{these}.

\medskip

Note that, if instead of supercells of the form $\Gamma_L = L \Gamma$, $L \in \N^\ast$, we had used computational domains of the form $\Gamma_{L+t} = (L+t) \Gamma$, $t \in (0,1)$, we would have observed spectral pollution. As in the case studied in the previous section, the spurious eigenvectors concentrate on the boundary $\partial \Gamma_{L+t}$. In the one-dimensional setting (${\mathcal R}=b\Z$), and for a fixed value of $t$, the translated spurious modes $\phi_{L,N}(\cdot-(L+t)b/2)$ strongly converge in $H^1_{\rm loc}(\R)$, when $L$ goes to infinity, to the normalized eigenmodes of the dislocation operator $H(t)=-\dps \frac{d^2}{dx^2} + 1_{x<0} V_{\rm per}(x+tb/2) +1_{x>0} V_{\rm per}(x-tb/2)$ studied in \cite{Korotyaev}. We refer to \cite{these} for further details.

\section{A no-pollution criterion}
\label{sec:projector}

Spectral pollution can be avoided by using e.g. the quadratic projection method, introduced in an abstract setting in \cite{QPM}, and applied to the case of perturbed periodic Schr\"odinger operators in~\cite{BoultonLevitin}. An alternative way to prevent spectral pollution is to impose constraints on the approximation spaces $(X_n)_{n \in \N}$. Consider a gap $(\alpha,\beta) \subset \R \setminus \sigma(H^0_{\rm per})$ in the spectrum of $H^0_{\rm per}$ and denote by $P=\chi_{(-\infty,\gamma]}(H^0_{\rm per})$ where $\gamma=\frac{\alpha+\beta}2$ and where $\chi_{(-\infty,\gamma]}$ is the characteristic function of the interval $(-\infty,\gamma]$. 

\medskip

\begin{theorem}\label{Th:Wannier} Let $(P_n)_{n \in \N}$ be a sequence of linear projectors on $L^2(\R^d)$ such that for all $n\in\N$, $\mbox{Ran}(P_n) \subset H^1(\R^d)$, and $\sup_{n \in \N}\|P_n\|_{{\mathcal L}(L^2)} < \infty$,
and $(X_n)_{n \in \N}$ a sequence of finite dimensional discretization spaces satisfying (\ref{eq:density}) as well as the following two properties:
\begin{description}
\item[(A1)] $\forall n \in \N, \; X_n = X_n^+ \oplus X_n^- \; \mbox{with} \; X_n^- \subset \mbox{\rm Ran}(P_n) \; \mbox{and} \; X_n^+ \subset \mbox{\rm Ran}(1-P_n)$;
\item [(A2)] $\dps \mathop{\sup}_{\phi_n \in X_n \setminus \left\{0\right\}} \frac{\|(P-P_n)\phi_n\|_{H^1(\R^d)}}{\|\phi_n\|_{H^1(\R^d)}} \mathop{\longrightarrow}_{n\to\infty} 0$.
\end{description}
Then,
$$
\lim_{n \to \infty} \sigma(H|_{X_n}) \cap (\alpha,\beta) = \sigma(H) \cap (\alpha,\beta).
$$
\end{theorem}

\medskip

The above result is an extension, for the specific case of perturbed periodic Schr\"odinger operators, to the results in~\cite[Theorem~2.6]{LewinSere} in the sense that (i) the exact spectral projector $P$ 
is replaced by an approximate projector $P_n$, and (ii) the discretization space $X_n$ may consist of functions of $H^1(\R^d)$ (the form domain of $H$), 
while in \cite{LewinSere}, the basis functions are assumed to belong to $H^2(\R^d)$ (the domain of $H$). 

\medskip

\begin{proof}
From (\ref{eq:no_lack}), we already know that $\sigma(H) \cap (\alpha,\beta) \subset \liminf_{n \to \infty} \sigma(H|_{X_n}) \cap (\alpha,\beta)$. Conversely, let $\lambda \in (\limsup_{n \to \infty} \sigma(H|_{X_n}) \cap (\alpha,\beta))\setminus \sigma(H)$, and $(\psi_{n_k})_{k\in\N}$ be a sequence of functions of $H^1(\R^d)$ such that for all $k\in\N$, $\psi_{n_k} \in X_{n_k}$, $\|\psi_{n_k}\|_{L^2(\R^d)} = 1$ and 
$(H|_{X_{n_k}} - \lambda)\psi_{n_k} \mathop{\longrightarrow}_{k\to\infty} 0$ strongly in $L^2(\R^d)$. Reasoning as in the proof of Proposition~\ref{prop:Galerkin}, we obtain that the sequence $(\psi_{n_k})_{k\in\N}$ converges to $0$,  weakly in $H^1(\R^d)$, and strongly in $L^2_{\rm loc}(\R^d)$. Let us then expand $\psi_{n_k}$ as $\psi_{n_k} = \psi_{n_k}^+ + \psi_{n_k}^-$ with $\psi_{n_k}^+:=(1-P_{n_k})\psi_{n_k} \in X_{n_k}^+$ and $\psi_{n_k}^- := P_{n_k}\psi_{n_k} \in X_{n_k}^-$ and notice that
$$
(a^0-\lambda)(\psi_{n_k}^+, \psi_{n_k}^+) + (a^0-\lambda)(\psi_{n_k}^-, \psi_{n_k}^+) = (a-\lambda)(\psi_{n_k}, \psi_{n_k}^+) - \int_{\R^d} W \psi_{n_k} \psi_{n_k}^+.
$$
Since $\psi_{n_k}^+=(1-P_n)\psi_{n_k} \in X_{n_k}$, 
\begin{eqnarray*}
\left|(a-\lambda)(\psi_{n_k}, \psi_{n_k}^+)\right| &=& \left|\langle (H|_{X_{n_k}}-\lambda)\psi_{n_k},(1-P_n)\psi_{n_k} \rangle_{L^2} \right| \\
&\le&  \left(1+\sup_{k \in \N}\|P_{n_k}\|_{{\mathcal L}(L^2)}\right)  \|(H|_{X_{n_k}}-\lambda)\psi_{n_k}\|_{L^2} \mathop{\longrightarrow}_{k \to \infty} 0.
\end{eqnarray*}
Besides, as $W$ vanishes at infinity, $(\psi_{n_k})_{k \in \N}$ converges to $0$ in $L^2_{\rm loc}(\R^d)$ and $\sup_{k \in \N} \|\psi_{n_k}^+\|_{L^2} \le 1+\sup_{k \in \N}\|P_{n_k}\|_{{\mathcal L}(L^2)} < \infty$, we also have
$$
\int_{\R^d} W \psi_{n_k} \psi_{n_k}^+ \mathop{\longrightarrow}_{k \to \infty} 0.
$$
Therefore,
$$
(a^0-\lambda)(\psi_{n_k}^+, \psi_{n_k}^+) + (a^0-\lambda)(\psi_{n_k}^-, \psi_{n_k}^+) \mathop{\longrightarrow}_{k\to\infty}  0.
$$
Likewise,
$$
(a^0-\lambda)(\psi_{n_k}^+, \psi_{n_k}^-) + (a^0-\lambda)(\psi_{n_k}^-, \psi_{n_k}^-)=(a-\lambda)(\psi_{n_k}, \psi_{n_k}^-) - \int_{\R^d} W \psi_{n_k} \psi_{n_k}^-  \dps \mathop{\longrightarrow}_{k\to\infty}  0.
$$
Substracting the second equation from the first one, we obtain
$$
(a^0-\lambda)(\psi_{n_k}^+, \psi_{n_k}^+) - (a^0-\lambda)(\psi_{n_k}^-, \psi_{n_k}^-) \mathop{\longrightarrow}_{k\to\infty} 0.
$$
Now, we notice that
\begin{eqnarray*}
(a^0-\lambda)(\psi_{n_k}^-, \psi_{n_k}^-) &= & (a^0-\lambda)(P_{n_k}\psi_{n_k}, P_{n_k} \psi_{n_k})\\
& = & (a^0-\lambda)(P\psi_{n_k}, P\psi_{n_k}) + 2(a^0-\lambda)(P\psi_{n_k}, (P_{n_k}-P)\psi_{n_k})\\
&& + (a^0-\lambda)((P_{n_k}-P)\psi_{n_k}, (P_{n_k}-P)\psi_{n_k}),
\end{eqnarray*}
and
\begin{eqnarray*}
(a^0-\lambda)(\psi_{n_k}^+, \psi_{n_k}^+) &= & (a^0-\lambda)((1-P_{n_k})\psi_{n_k}, (1-P_{n_k}) \psi_{n_k}) \\
& = & (a^0-\lambda)((1-P)\psi_{n_k}, (1-P)\psi_{n_k}) \\
&& + 2(a^0-\lambda)((1-P)\psi_{n_k}, (P-P_{n_k})\psi_{n_k})\\
&& + (a^0-\lambda)((P-P_{n_k})\psi_{n_k}, (P-P_{n_k})\psi_{n_k}).
\end{eqnarray*}
Besides, there exists $\eta_+, \eta_->0$ such that for all $\psi\in H^1(\R^d)$, 
\begin{eqnarray*}
(a^0-\lambda)((1-P)\psi, (1-P)\psi) &\geq & \eta_+ \|(1-P)\psi\|_{L^2(\R^d)}^2, \\
-(a^0-\lambda)(P\psi, P\psi) & \geq & \eta_- \|P\psi\|_{L^2(\R^d)}^2.
\end{eqnarray*}
Thus, 
\begin{eqnarray*}
(a^0-\lambda)(\psi_{n_k}^+, \psi_{n_k}^+) - (a^0-\lambda)(\psi_{n_k}^-, \psi_{n_k}^-) & \geq & \min(\eta_+, \eta_-) \|\psi_{n_k}\|_{L^2(\R^d)}^2 \\ &&
+ 2(a^0-\lambda)(\psi_{n_k}, (P-P_{n_k})\psi_{n_k}).
\end{eqnarray*}
From assumption $(A2)$ and the boundedness of $(\psi_{n_k})_{k \in \N}$ in $H^1(\R^d)$, we deduce that
$$
(a^0-\lambda)(\psi_{n_k}, (P-P_{n_k})\psi_{n_k}) \mathop{\longrightarrow}_{k\to\infty} 0,
$$
which imply that $\dps \|\psi_{n_k}\|_{L^2} \mathop{\longrightarrow}_{k\to\infty} 0$. This contradicts the fact that $\|\psi_{n_k}\|_{L^2}=1$ for all $k \in \N$. 
\end{proof}

\medskip

The assumptions made in Theorem~\ref{Th:Wannier} allow in particular to consider approximation spaces built from approximate spectral projectors of $H^0_{\rm per}$. As a matter of illustration, let us consider the case when the approximate spectral projectors are constructed by means of the finite element method. As in Section~\ref{sec:spectral_pollution}, we consider a sequence $({\cal T}_n^\infty)_{n \in \N}$ of uniformly regular meshes of $\R^d$, invariant with respect to the translations of the lattice ${\cal R}$, and such that $h_n:=\max_{K \in {\cal T}_n^\infty} \mbox{diam}(K) \mathop{\longrightarrow}_{n \to \infty}0$, and denote by $X_n^\infty$ the infinite dimensional closed vector subspace of $H^1(\R^d)$ built from $({\cal T}_n^\infty)_{n \in \N}$ and $\P_m$-finite elements. Assume that we want to compute the eigenvalues of $H=H^0_{\rm per}+W$ located inside the gap $(\alpha,\beta)$ between the $J^{\rm th}$ and $(J+1)^{\rm st}$ bands of $H^0_{\rm per}$. Using Bloch theory \cite{ReedSimon4}, we obtain
$$
P = \chi_{(-\infty,\gamma]}(H^0_{\rm per}) = \fint_{\Gamma^\ast} P_q \, dq,
$$
where $P_q$ is the rank-$J$ orthogonal projector on $L^2_q(\Gamma)$ defined by
$$
P_q = \sum_{j=1}^J |e_{j,q}\rangle \, \langle e_{j,q}|,
$$
where $(\epsilon_{j,q},e_{j,q})_{j \in \N^\ast}$, $\epsilon_{1,q} \le \epsilon_{2,q} \le \cdots$, is an $L^2_q(\Gamma)$-orthonormal basis of eigenmodes of the quadratic form $a^0_q$ defined by (\ref{eq:a0q}). For $n$ large enough, we introduce 
\begin{equation} \label{eq:defPn}
P_n := \fint_{\Gamma^\ast} \sum_{j=1}^J |e_{n,j,q}\rangle \, \langle e_{n,j,q}| \, dq,
\end{equation}
where $(\epsilon_{n,j,q},e_{n,j,q})_{1 \le j \le N_n}$, $\epsilon_{n,1,q} \le \epsilon_{n,2,q} \le \cdots \le \epsilon_{n,N_n,q}$, is the $L^2_q(\Gamma)$-orthonormal basis of eigenmodes of $a_q^0$ in $C^0({\mathcal T}_n^\infty) \cap L^2_q(\Gamma)$ already introduced in the proof of Proposition~\ref{prop:Galerkin}.

\medskip

\noindent
We have seen in Section~\ref{sec:spectral_pollution} that using approximation spaces of the form
$$
X_n = \left\{\psi_n \in X_n^\infty \; | \; \mbox{Supp}(\psi_n) \subset \Omega_n \right\},
$$
where $(\Omega_n)_{n \in \N}$ is an increasing sequence of closed convex sets of $\R^d$ converging to $\R^d$, leads, in general, to spectral pollution. We now consider the approximation spaces
\begin{equation} \label{eq:augmented}
\widetilde X_n = X_n^+ \oplus X_n^- \quad \mbox{where} \quad X_n^-=P_nX_n\quad \mbox{and} \quad X_n^+ = (1-P_n)X_n.
\end{equation}
Note that $\widetilde X_n= X_n + P_n X_n$, so that $\widetilde X_n$ can be seen as an augmentation of $X_n$.

\medskip

\begin{corollary} \label{Cor:Wan}
The sequence of approximation spaces $(\widetilde X_n)_{n \in \N}$ defined by (\ref{eq:augmented}) satisfies (\ref{eq:density}) and it holds
\begin{equation} \label{eq:wan}
\lim_{n \to \infty} \sigma(H|_{\widetilde X_n}) \cap (\alpha,\beta) = \sigma(H) \cap (\alpha,\beta).
\end{equation}
\end{corollary}

\medskip

\begin{proof} As $\widetilde X_n= X_n + P_n X_n$ with $(X_n)_{n \in \N}$ satisfying (\ref{eq:density}), it is clear that $(\widetilde X_n)_{n \in \N}$ satisfies (\ref{eq:density}). The sequence $(P_n)_{n \in \N}$ is a sequence of orthogonal projectors of $L^2(\R^d)$ such that $\mbox{Ran}(P_n) \subset X_n^\infty \subset H^1(\R^d)$. Besides, $\|P_n\|_{{\mathcal L}(L^2)}=1$ since the projector $P_n$ is orthogonal.  It follows from the minmax principle~\cite{ReedSimon4} and usual a priori error estimates for linear elliptic eigenvalue problems \cite{linear} that  
$$
\sup_{1 \le j \le J, \, q \in \Gamma^\ast} \epsilon_{n,j,q} \mathop{\longrightarrow}_{n \to \infty} \alpha \qquad \mbox{and} \qquad \inf_{j \ge J+1, \, q \in \Gamma^\ast} \epsilon_{n,j,q} \mathop{\longrightarrow}_{n \to \infty} \beta,
$$
and that there exists $C \in \R_+$ such that
$$
\| P_n-P \|_{{\mathcal L}(H^1)} \le C \, \sup_{q \in \Gamma^\ast} \sup_{\begin{array}{l} v_q \in {\rm Ran}(P_q) \\ \|v_q\|_{L^2_q(\Gamma)}=1 \end{array}} \inf_{v^n_q \in  C^0({\mathcal T}_n^\infty) \cap L^2_q(\Gamma)} \|v_q-v^n_q\|_{H^1_q(\Gamma)} \mathop{\longrightarrow}_{n \to \infty} 0.
$$
We conclude using Theorem~\ref{Th:Wannier}.
\end{proof}

\medskip

\noindent
Let us finally present some numerical simulations illustrating Corollary~\ref{Cor:Wan} in a one-dimensional setting, with $V_{\rm per}(x) = \cos(x) + 3\sin(2x+1)$ and $W(x) = -(x+2)^2 e^{-x^2}$. We focus on the spectral gap $(\alpha,\beta)$ located between the first and second bands of $H^0_{\rm per}=-\frac{d^2}{dx^2}+V_{\rm per}$ (corresponding to $J=1$). Numerical simulations done with the pollution-free supercell model show that $\alpha \simeq -1.15$ and $\beta \simeq -0.65$, and that $H$ has exactly two discrete eigenvalues $\lambda_1 \simeq -1.04$ and $\lambda_2 \simeq -0.66$ in the gap $(\alpha,\beta)$.

The simulations below have been performed with a uniform mesh of $\R$ centered on $0$, consisting of segments of length $h = \pi/50$, and with $\Omega = [-L,L]$, for different values of $L$. The sums over ${\mathcal R}$ have been truncated using very large cut-offs; likewise, the integrals on the Brillouin zone have been computed numerically on a very fine uniform integration grid, in order to eliminate the so-called $k$-point discretization errors. The numerical analysis of the approximations resulting from the truncation of the sums over ${\mathcal R}$ and from the numerical integration on $\Gamma^\ast$, is work in progress.

The spectra of the operators $H|_{X_n}$ (standard finite element discretization spaces) and $H|_{\widetilde X_n}$ (augmented finite element discretization spaces defined by (\ref{eq:augmented})) are displayed in Figure~\ref{fig:Wannier}. The variational approximation of $H$ in $X_n$ is seen to generate spectral pollution, while, in agreement with Corollary~\ref{Cor:Wan}, no spectral pollution is observed with the discretization spaces $\widetilde X_n$.

\medskip

\begin{figure}[h]
\centering
\label{fig:Wannier}
\begin{tabular}{c}
\includegraphics[height=5.5truecm]{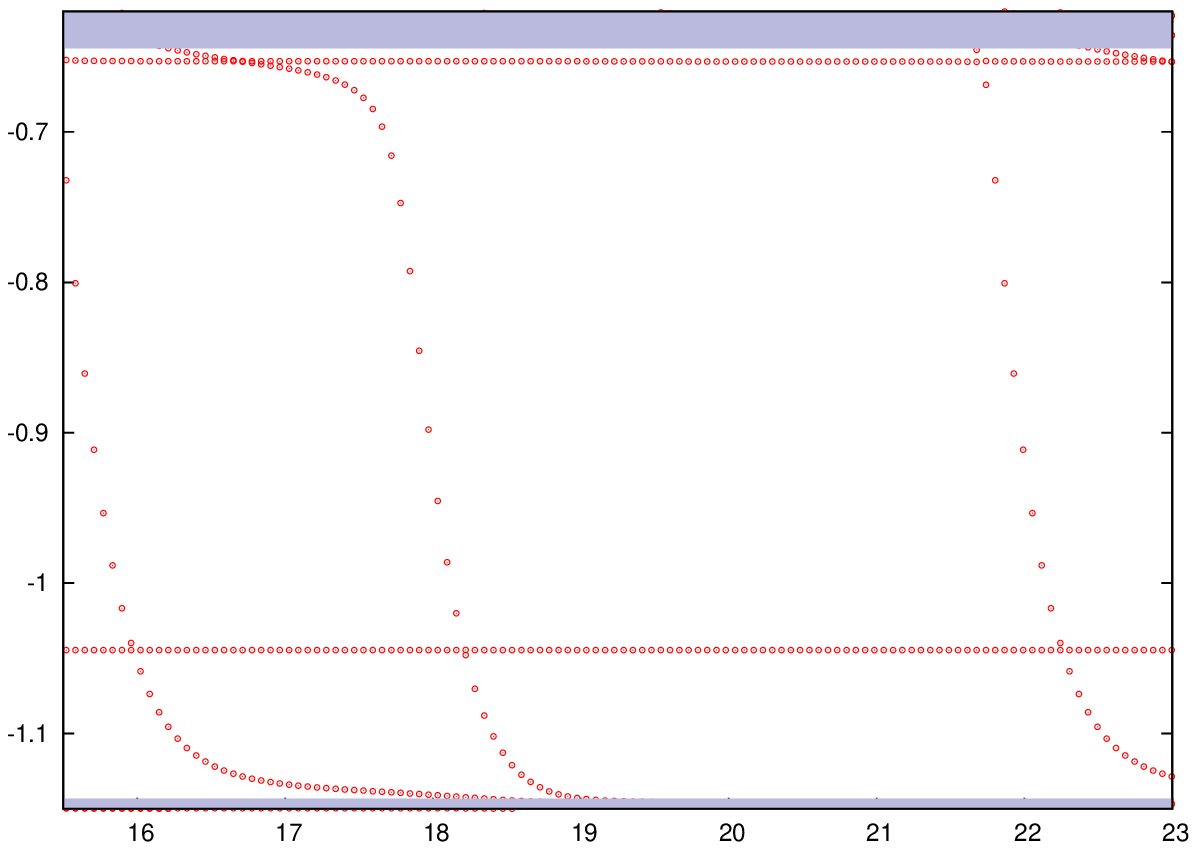} \\
\includegraphics[height=5.5truecm]{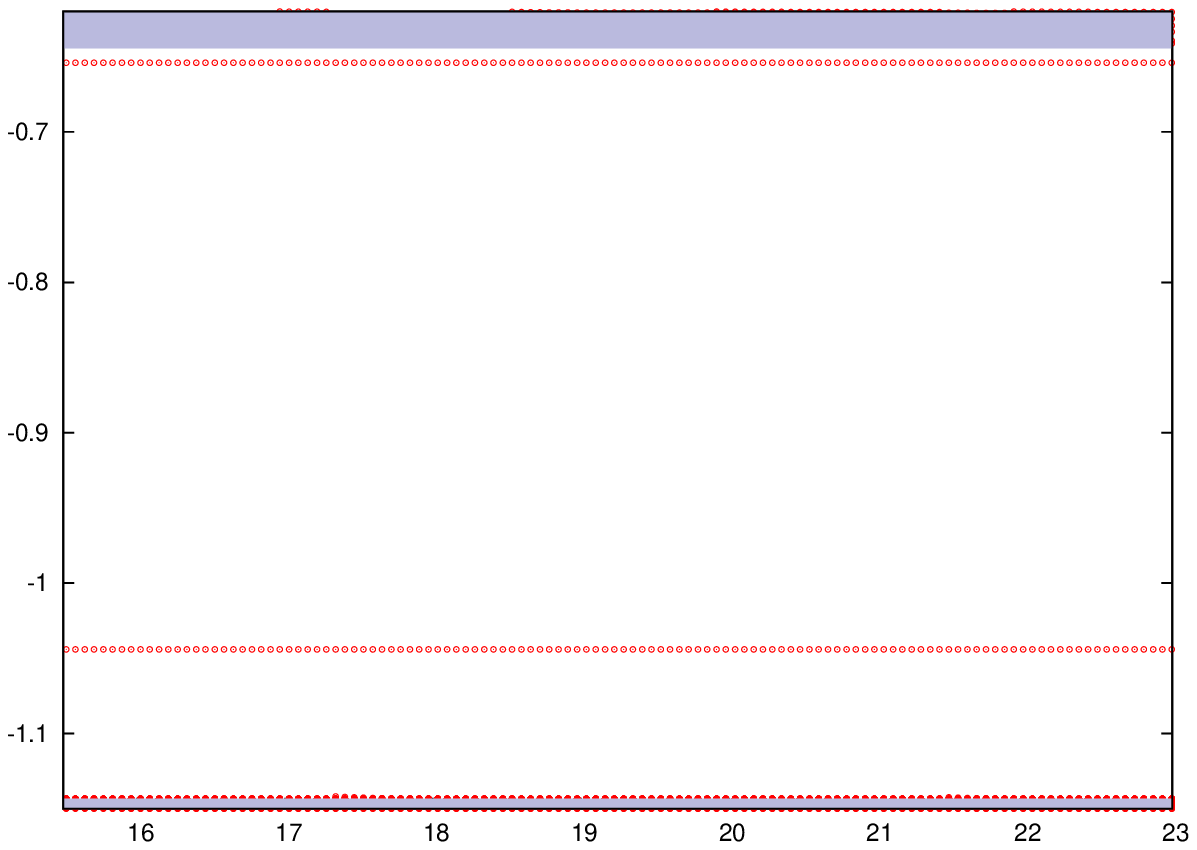}
\end{tabular}
\caption{The spectra of the variational approximations of $H$ for various sizes of the simulation domain, obtained with standard finite element discretization spaces $X_n$ (top) and with augmented finite element discretization spaces $\widetilde X_n$ defined by (\ref{eq:augmented}) (bottom).}
\end{figure}

\bigskip

\section*{Acknowledgements} We thank Fran\c cois Murat for helpful discussions.

\end{document}